\documentclass[11pt]{article}

\usepackage{fullpage}
\usepackage{mathtools}

\usepackage{hyperref}       %
\usepackage{url}            %
\usepackage{amsfonts}       %
\usepackage{xcolor}         %

\title{Bayesian Perspective on Memorization and Reconstruction}

\author{%
Haim Kaplan\thanks{Tel Aviv University and Google Research. Supported in part by ISF grant 1156/23 and the Blavatnik Research Foundation.}
\and
Yishay Mansour\thanks{Tel Aviv University and Google Research. This project has received funding from the European Research Council (ERC) under the European Union’s Horizon 2020 research and innovation program (grant agreement No. 882396), by the Israel Science Foundation, the Yandex Initiative for Machine Learning at Tel Aviv University and a grant from the Tel Aviv University Center for AI and Data Science (TAD).}
\and
Kobbi Nissim\thanks{Georgetown University and Google Research. Work is partially funded by NSF grant No.\ 2217678 and by a gift to
Georgetown University.}
\and
Uri Stemmer\thanks{Tel Aviv University and Google Research. Supported in part by ISF grant 1419/24 and the Blavatnik Research Foundation.}
}

\allowdisplaybreaks

\newcommand{\remove}[1]{}

\usepackage{stackrel}

\usepackage[most]{tcolorbox}
\usepackage{caption}

\usepackage{natbib}
\usepackage{tcolorbox}
\usepackage{algorithm}

\usepackage{scalerel}
\usepackage{mathrsfs}
\usepackage{enumitem}
\usepackage{framed}
\usepackage{bbm}
\usepackage{tikz}
\usepackage{float}
\usetikzlibrary{backgrounds}
\usepackage{color}
\usepackage{graphicx}
\usepackage{latexsym}
\usepackage{amsfonts}
\usepackage{amsmath, amssymb, amsthm}
\usepackage{multirow}
\usepackage{dsfont}
\usepackage{array}
\usepackage{hyperref}
\usepackage{adjustbox}

\newcommand{\smallminus}{\scalebox{0.75}[1.0]{$-$}}

\newcommand{\BN}
{\text{Bayesian extraction-safe}}

\def\restrict#1{\raise-.5ex\hbox{\ensuremath|}_{#1}}

\DeclareSymbolFont{AMSb}{U}{msb}{m}{n}
\DeclareMathSymbol{\N}{\mathbin}{AMSb}{"4E}
\DeclareMathSymbol{\Z}{\mathbin}{AMSb}{"5A}
\DeclareMathSymbol{\R}{\mathbin}{AMSb}{"52}
\DeclareMathSymbol{\Q}{\mathbin}{AMSb}{"51}
\DeclareMathSymbol{\erert}{\mathbin}{AMSb}{"50}
\DeclareMathSymbol{\I}{\mathbin}{AMSb}{"49}
\DeclareMathSymbol{\C}{\mathbin}{AMSb}{"43}

\newcommand{\mynote}[2]{{\textcolor{#1}{ #2}}}
\definecolor{gray}{gray}{0.4}
\newcommand{\gray}[1]{\mynote{gray}{{\scriptstyle #1}}}

\newtheorem{theorem}{Theorem}[section]

\newtheorem{lemma}[theorem]{Lemma}

\newtheorem{definition}[theorem]{Definition}

\newtheorem{remark}[theorem]{Remark}

\newtheorem{proposition}[theorem]{Proposition}
\newtheorem{claim}[theorem]{Claim}

\newtheorem{corollary}[theorem]{Corollary}

\newtheorem{observation}[theorem]{Observation}

\newtheorem{example}[theorem]{Example}

\newcommand{\1}{\mathbbm{1}}

\newcommand{\AAA}{\mathcal A}

\newcommand{\DDD}{\mathcal D}

\newcommand{\MMM}{\mathcal M}

\newcommand{\XXX}{\mathcal X}

\newcommand{\eps}{\varepsilon}

\newcommand{\z}{\mathrm{z}}

\newcommand{\Lap}{\operatorname{\rm Lap}}
\newcommand{\Hamming}{\operatorname{\rm Ham}}
\newcommand{\Ber}{\operatorname{\rm Ber}}
\newcommand{\SideInfo}{\operatorname{\rm SideInfo}}
\newcommand{\Nature}{\operatorname{\rm Nature}}
\newcommand{\Prior}{\operatorname{\rm Prior}}
\newcommand{\Tardos}{\operatorname{\rm Tardos}}

\newcommand{\average}{\operatorname{\rm Average}}

\newcommand{\poly}{\mathop{\rm poly}}

\newcommand{\noise}{\operatorname{\tt noise}}
\def\E{\operatorname*{\mathbb{E}}}

\def\Q{\operatorname*{\mathbb{Q}}}
\def\poly{\mathop{\rm{poly}}\nolimits}
\def\Lap{\mathop{\rm{Lap}}\nolimits}

\newcommand{\Avg}{\operatorname{\rm Average}}

\makeatletter
\newcommand{\subalign}[1]{%
  \vcenter{%
    \Let@ \restore@math@cr \default@tag
    \baselineskip\fontdimen10 \scriptfont\tw@
    \advance\baselineskip\fontdimen12 \scriptfont\tw@
    \lineskip\thr@@\fontdimen8 \scriptfont\thr@@
    \lineskiplimit\lineskip
    \ialign{\hfil$\m@th\scriptstyle##$&$\m@th\scriptstyle{}##$\hfil\crcr
      #1\crcr
    }%
  }%
}
\makeatother

\begin{document}

\date{May 29, 2025}
\maketitle

\begin{abstract}
\noindent We introduce a new Bayesian perspective on the concept of data reconstruction, and leverage this viewpoint to propose a new security definition that, in certain settings, provably prevents reconstruction attacks. We use our paradigm to shed new light on one of the most notorious attacks in the privacy and memorization literature -- fingerprinting code attacks (FPC). We argue that these attacks are really a form of {\em membership inference} attacks, rather than {\em reconstruction} attacks. Furthermore, we show that if the goal is solely to prevent reconstruction (but not membership inference), then in some cases the impossibility results derived from FPC no longer apply. 
\end{abstract}

\section{Introduction}

\citet{Carlini0EKS19} showed that it is sometimes possible to {\em extract} unique pieces of training data from modern language models (such as credit card numbers). This demonstrates that such models can unintentionally {\em memorize} rare parts of their training data, even if those parts appear only once. Since then, this {\em memorization phenomenon} has been studied in a long line of work, providing increasingly many examples in which modern models unintentionally memorize data. In fact, several follow-up papers have shown that there exist learning tasks for which memorization is {\em provably necessary} \citep{Feldman20,FeldmanZ20,carlini2021extracting,BrownBFST21,HaimVYSI22,BuzagloHYVONI23,CarliniHNJSTBIW23,CarliniIJLTZ23}.

However, these prior works did not converge on a single {\em definition of memorization}, and instead considered several context-dependent notions. Loosely speaking, existing definitions can be grouped into two main types: 

\begin{itemize}
\item {\bf Computational memorization.} Some papers, such as \citet{Carlini0EKS19}, define {\em memorization} in terms of {\em extraction}. Informally, a model $h$ is said to ``memorize'' a training example $x$ if $x$ can be extracted from $h$. This is in line with the cryptographic concept of {\em knowledge extraction} in proofs of knowledge: To demonstrate that Alice ``knows'' a fact there should be a way to extract that fact from her.

\item {\bf Statistical memorization.} Other works define {\em memorization} using statistical or information-theoretic criteria. For example, a model $h$ is said to ``memorize'' a training example $x$ if the mutual information between $h$ and $x$ is high. This is in line with the information-theoretic view of knowledge: A system is considered to possess information about a variable if observing the system reduces uncertainty about that variable.

\end{itemize}

In light of this, we ask the following questions:

\begin{tcolorbox}[boxrule=1pt]
\begin{center}
Is there a difference between {\em computational memorization} and {\em statistical memorization}? Could a model be considered ``memorizing'' under one definition but not the other?
\end{center}
\end{tcolorbox}

In this paper, we provide positive answers to the above questions. Specifically, we show that there are {\em well-established} examples of statistical memorization in which ``extraction'' is {\em provably impossible}. Specifically, 
\begin{enumerate}
    \item We put forward a new security definition, computational in nature, such that algorithms that satisfy it are guaranteed to resist extraction attacks.  
    \item We revisit problems studied in prior work for when statistical-memorization is provably necessary, and show that these problems can be solved while satisfying our new security notion (under the same conditions as studied in prior work). Therefore these examples separate statistical from computational memorization.
\end{enumerate}

This raises the question of whether ``statistical memorization'' should be considered as a form of ``memorization'' at all. Intuitively, if it is impossible for an attacker to {\em extract} information from the model, then had the model really {\em memorized} it?

\subsection{Existing definitions of statistical memorization}

Before describing our new results more formally, we survey some of the existing statistical definitions for memorization. 
\citet{BrownBFST21} defined memorization in terms of {\em conditional mutual information} as follows.

\begin{definition}\label{def:statMemo1}[\citet{BrownBFST21}]
    Let $\Prior$ be a meta-distribution (i.e., a distribution over distributions). Let $\Nature$ be a target distribution sampled from $\Prior$, and let $S$ be a sample from $\Nature$. Let $h\leftarrow\MMM(S)$ be a model obtained by running a learning algorithm $\MMM$ on the sample $S$. The amount of information that $h$ memorizes about $S$ is defined as $I(S ; h | \Nature)$.

\end{definition}

In this definition, the meta-distribution $\Prior$ is assumed to be known, but the specific target distribution $\Nature$ is a priori unknown. To interpret this definition, consider a learning algorithm $\MMM$ that given a sample $S$ is trying to learn properties of the underlying distribution $\Nature$, say learn a good predictor $h$. Now note that $I(S ; h | \Nature)$ captures the amount of information that $h$ encodes about $S$ {\em after excluding information about the target distribution itself}. Intuitively, this means that $I(S ; h | \Nature)$ captures the amount of information encoded in $h$ that is ``unique to the dataset''.

\citet{AttiasDHL024} considered a different definition for memorization, stated in terms of {\em membership inference attacks (MIA)} \citep{ShokriSSS17,yeom2018privacy}. Membership inference is a type of attack in which, given a model $h$ and a data point $x$, the attacker aims to determine whether $x$  was part of the training data underlying the model $h$. Formally,

\begin{definition}\label{def:statMemo2}[\citet{ShokriSSS17,yeom2018privacy,AttiasDHL024,VoitovychHAKDLR25}]
Let $\XXX$ be a data domain. 
Let $\MMM$ be an algorithm whose input is a dataset containing $n$ elements from $\XXX$. Algorithm $\MMM$ is {\em $(m\in\N,\xi\in(0,1))$-memorizing} if there is a distribution $\DDD$ over $\XXX$ and a distinguisher $\AAA$ such that
$$
\Pr_{\substack{S\leftarrow\DDD^n\\
    h\leftarrow\MMM(S)\\
    x\leftarrow S\\
    b\leftarrow\AAA(h,x)}}[b=\textsc{``In''}]\geq\frac{m}{n}
    \qquad\text{and}\qquad
\Pr_{\substack{S\leftarrow\DDD^n\\
    h\leftarrow\MMM(S)\\
    x\leftarrow\DDD\\
    b\leftarrow\AAA(h,x)}}[b=\textsc{``In''}]\leq\xi
    $$
\end{definition}

In words, the attacker $\AAA$ is executed on $h\leftarrow\MMM(S)$ and on a data point $x\in\XXX$. If this $x$ is sampled uniformly from $S$, then there is a noticeable probability (at least $m/n$) that $\AAA$ correctly determines that $x\in S$. Intuitively this happens if $h$ ``memorizes'' $m$ out of the $n$ elements of $S$. On the other hand, when $x$ is instead sampled independently from $S$, then the probability that $\AAA$ falsely determines that $x\in S$ should be small (at most $\xi$).

\subsection{Warmup: Memorization can be ambiguous}

In this paper we show that there are {\em well-established} cases where memorization is ``formally ambiguous''. Before presenting these results, we motivate our discussion with an intuitive example for a mechanism whose memorization behavior is open to interpretation. 

\begin{example}\label{ex:xor}
Let $n,d\in\N$ be parameters where $d$ is even, let $\Nature$ be the uniform distribution over $\{0,1\}^d$, and let $S=(x_1,\dots,x_n)\leftarrow\Nature^n$ be a sample containing $n$ iid elements from $\Nature$. We think of $S$ as an $n\times d$ binary matrix. Consider the mechanism $\MMM$ that, given $S$, returns an $n\times \frac{d}{2}$ matrix $H$ where $H[i,j]=S[i,2j-1]\oplus S[i,2j]$.
\end{example}

In this example, we have that $I(S;H | \Nature)=nd/2$ (intuitively, this holds because given $H$ the number of choices for every pair of bits in $S$ drops from 4 to 2). So $\MMM$ memorizes a lot of information about $S$ according to Definition~\ref{def:statMemo1}. 
In addition, there is a trivial membership inference attacker $\AAA$ against this mechanism. Specifically, given the release $H$ and a row $x\in\{0,1\}^d$, the attacker $\AAA$ answers \textsc{``In''} iff $x$ fits one of the rows in $H$, i.e., there is a row $i\in[n]$ such that for every coordinate $j\in[d/2]$ we have $H[i,j]=x_{2j-1}\oplus x_{2j}$. Note that for every $x\in S$ we have that $\AAA(H,x)=\textsc{``In''}$. But
for a random $x\in\{0,1\}^d$ we have that $\Pr[\AAA(H,x)=\textsc{``In''}]\leq n\cdot 2^{-d/2}$. So $\MMM$ essentially ``memorizes'' all elements in $S$ according to Definition~\ref{def:statMemo2}; formally it is $(n,n{\cdot}2^{-d/2})$-memorizing. 

However, even given the release $H$, the marginal distribution on any single bit in $S$ is still completely uniform. Furthermore, even given $H$, no attacker  (even unbounded) can guess noticeably more than half of the coordinates of any single row in $S$. Therefore, while there is a strong sense in which $\MMM$ ``memorizes'' a lot of information (according to Definitions~\ref{def:statMemo1} and~\ref{def:statMemo2}), there is also a strong sense in which no attacker can extract training examples from $H$. It is therefore open to interpretation whether this should be considered memorization or not.

\subsection{Our security definition: arguing about the attacker's knowledge is crucial}

Example~\ref{ex:xor} is simple enough so that we could argue that ``extraction is impossible'' without formally defining what it means. This becomes much less clear when dealing with more involved algorithms/problems. 
In order to tackle this, we first introduce a formal security definition for preventing extraction, and then present algorithms satisfying our definition. Building up towards our security definition, we now argue that such a definition must take into account the attacker's side-information (or lack thereof) about the sample and/or the underlying distribution.

\begin{example}\label{ex:xor2}
Consider again Example~\ref{ex:xor}, where the mechanism $\MMM$ takes a dataset $S$ containing $nd$ random bits and releases $\frac{nd}{2}$ ``parity bits'' (denoted as $H$). As we mentioned, in this scenario no attacker can guess noticeably more than half of the bits of any single row in $S$.  
However, this argument clearly breaks if the attacker possesses enough side-information about the sample $S$. For example, suppose that the attacker knows a superset $\mathbb{S}$ of size $n^{10}$  that contains $S$ as well as additional $(n^{10} \hspace{1px}\smallminus\hspace{1px} n)$ independent samples from $\{0,1\}^d$. Now, if the attacker knows $\mathbb{S}$ and sees the release $H$, then it can easily identify all rows in $S$ with high confidence. Specifically, the attacker searches for all rows in $\mathbb{S}$ that fit perfectly to a row in $H$. This identifies all rows in $S$, and the probability of falsely identifying a row outside of $S$ is bounded by $2^{-d/2}\cdot\poly(n)$.    
\end{example}

But did the attacker really {\em extract} $S$ from the release $H$, or was it already encoded in the side-information available to it? 
 Taking this example to the limit, suppose that the attacker's side-information explicitly specifies a vector $x^*$ from the sample $S$. Here it is intuitively clear that if this attacker sees the release $H$ and then ``identifies'' that $x^*\in S$, then this is not a very successful attack as this information was {\em not} extracted from the release $H$ in any way. An alternative way to interpret this, is that we cannot hope to defend against ``extracting'' $x^*$ by attackers who already ``know'' that $x^*\in S$. What about attackers who know that there is a $99\%$ chance that $x^*\in S$? What about $10\%$? This suggests that we need some way of {\em quantifying} the amount of side-information the attacker possessed {\em before seeing the release} on the element that it chooses to reconstruct. Furthermore, we argue that there are cases where limiting the amount of side-information the attacker has about $S$ is insufficient, and we must also take into consideration the amount of side-information it has about the {\em underlying data distribution}. 

\begin{example}
Let us consider the following variant of Example~\ref{ex:xor2}. Suppose that the data distribution $\Nature$ is generated at random, by first randomly sampling a set $\mathbb{S}\subseteq\{0,1\}^d$ of size $n^{10}$, and then defining $\Nature$ to be uniform over $\mathbb{S}$. Now a sample $S$ of size $n$ is sampled from $\Nature$ and we compute $H\leftarrow\MMM(S)$, where $\MMM$ is the same mechanism from the previous two examples. As in Example~\ref{ex:xor2}, an attacker that knows $\mathbb{S}$ and sees the release $H$ can easily reconstruct all of $S$. But now the side information that the attacker possesses is more about the underlying distribution $\Nature$ rather than about the sample $S$, in the sense that conditioned on $\Nature$ the attacker does not hold any side-information.
\end{example}

This is actually reminiscent of the way some of the negative results are stated in the literature (for showing that memorization is necessary), e.g., by \citet{AttiasDHL024}: The attacker possesses a lot of side information;  not directly on the sample $S$, but rather through knowledge on the underlying distribution $\Nature$. 
Motivated by this, we advocate that in order for a ``reconstruction attack'' to be successful, the attacker needs to reconstruct an element $z$ which it had ``little'' or ``bounded'' knowledge about before seeing the release. 

In particular, this means that we need a way to argue about {\em lack of knowledge} in order to be able to designate an attack as successful. We model this via a Bayesian perspective: We assume that the underlying data distribution $\Nature$ is itself a random element, sampled from a family of possible data distributions $\Prior$. The dataset $S$ is then sampled from $\Nature$. A data analysis mechanism $\MMM$ is then executed on $S$ and outputs a release $H$. An attacker $\AAA$ who knows the meta-distribution $\Prior$ (but not $\Nature$ and $S$ directly) then sees the release $H$ and tries to extract information about $S$ form $H$. We will denote the output of $\AAA$ as $z$. 
Treating both $\Nature$ and $S$ as random elements allows us to model the ``lack of knowledge'' that the attacker has about them. 
We will say that the attacker is successful only if it does not have a lot of prior knowledge, through $\Nature$ and $S$, about $z$.

{\bf Ruling out trivial attacks.\;} 
Suppose for now that the attacker aims to identify an element $z$ such that $z\in S$ (we will later consider more general types of attacks). Consider the mechanism $\MMM_{\rm bot}$ that takes a dataset $S$, ignores it, and returns $\bot$. Clearly, this mechanism should be considered secure. However, as we assume that the attacker knows $\Prior$, it could be aware of ``trivial'' elements which appear in the sample $S$ w.h.p.\ even over sampling $\Nature$ from $\Prior$. Then, paradoxically, the attacker can ``extract'' these trivial elements from the empty release $\bot$. 

\begin{example}
Suppose that the meta-distribution $\Prior$ is defined via the following process for sampling $\Nature$. Sample a set $\mathbb{S}\subseteq\{0,1\}^d$ of size $n^{10}$ and return the distribution $\Nature$ that with with probability $0.5$ returns $\vec{0}$ and with probability $0.5$ returns a uniform element from $\mathbb{S}$. As the attacker knows $\Prior$, it could simply guess that $\vec{0}\in S$.
\end{example}

Motivated by this, in Sections~\ref{sec:vanilla} and~\ref{sec:sideInfo} we present a security definition, called {\em Bayesian Extraction-Safe}, that (informally) prevents an attacker from recovering an element $z\in S$ unless: (1) The attacker possessed ``too much'' side-information about $z$ before seeing the release; or
    (2) This $z$ is a ``trivial'' element which could be ``reconstructed'' even without seeing any release.

Our security definition builds on a prior security definition by \cite{CohenKMMNST25}, called {\em Narcissus resiliency}, and can be viewed as a Bayesian extension of it capturing ``lack of knowledge'' in the eyes of the attacker. 
We stress that our proposed definition is {\em not a privacy definition}. In particular, it does not prevent the attacker from learning potentially sensitive bits of information about the sample. The goal of our security definition is to protect against {\em extraction}, which we interpret as reconstructing large portions from the dataset or from a row in the dataset (this will be made formal in the actual definition). In other words, while a reconstruction attack clearly indicates a privacy breach, the other direction is not necessarily true: Not all privacy attacks should be considered as successful reconstruction attacks.

\subsection{Circumventing fingerprinting code attacks}

We use our security definition to shed new light on one of the most notorious attacks in the privacy and memorization literature -- fingerprinting code attacks (FPC). We argue that these attacks are really a form of {\em membership inference} attacks, rather than {\em reconstruction} attacks. Furthermore, we show that if the goal is solely to prevent reconstruction (but not membership inference), then in some cases the impossibility results derived from FPC no longer apply.

Specifically, we revisit the Tardos FPC, introduced by \citet{Tardos03}. This is one of the most influential and applicable FPC in the literature. In particular, this FPC is exactly what underlies the results of \citet{AttiasDHL024} who showed that there are Stochastic Convex Optimization (SCO) problems which necessitate memorization (in the terminology of Definition~\ref{def:statMemo2}). Our results show that under the same conditions studied by \citet{AttiasDHL024}, for which they proved that statistical-memorization is necessary, there actually exists a secure mechanism (satisfying our security definition) that provably prevents extraction. Therefore, our results show that with a different viewpoint, it becomes open to interpretation whether the results of \citet{AttiasDHL024} show that the SCO problems they studied really necessitate ``memorization'', or solely forces the algorithm to be susceptible to membership inference attacks.

\section{Preliminaries}

\cite{CohenKMMNST25} introduced a definition, called Narcissus resiliency, that aims to capture what it means for a system to be secure against reconstruction attacks. As ``reconstruction'' is context-dependent, the definition of Narcissus resiliency introduces a relation $R$, where $R(S, z) = 1$ if and only if $z$ is a ``valid reconstruction'' of $S$ (or of some data point in $S$). A simple example to keep in mind is $R(S,z) \triangleq \1_{\{z\in S\}}$.

To motivate the definition of Narcissus resiliency, consider the relation $R(S,z) \triangleq \1_{\{z\in S\}}$ and consider a distribution over images, such that every large enough dataset sampled from this distribution contains the Mona Lisa. In this case, if an attacker ``reconstructs'' the Mona Lisa from the outcome of our mechanism, then this is not a very meaningful attack, as this could also be done without accessing the outcome of our mechanism. The challenge is to distinguish between reconstructing something ``specific'' to the particular dataset versus reconstructing something generic. The key idea in the definition of Narcissus resiliency is its self-referential nature. Instead of trying to formally capture what is generic and what is not (which is very nuanced), the following definition places this decision on the attacker itself. Intuitively, the definition requires the attacker to come up with a valid reconstruction w.r.t.\ the actual dataset $S$ which would {\em not} be considered a valid reconstruction w.r.t.\ a fresh dataset $T$. The formal definition follows.

\begin{definition}[\cite{CohenKMMNST25}]\label{def:narc}
Let $\XXX$ be a data domain, let $\DDD$ be a distribution over datasets containing elements from $\XXX$, and let $R:\XXX^*\times\{0,1\}^*\rightarrow\{0,1\}$. Algorithm $\MMM$ is $(\eps,\delta,\DDD,R)$-Narcissus-resilient if for all attackers $\AAA$ it holds that
\begin{equation*}%
\underset{\substack{S\leftarrow\DDD\\y\leftarrow\MMM(S)\\z\leftarrow\AAA(y)}}{\Pr}[R(S,z)=1]\leq e^{\eps}\cdot \underset{\substack{S\leftarrow\DDD\\T\leftarrow\DDD\\y\leftarrow\MMM(S)\\z\leftarrow\AAA(y)}}{\Pr}[R(T,z)=1]+\delta.
\end{equation*}
\end{definition}

\paragraph{Notations.} %
For two multisets $A,B$ we have that $A\setminus B$ is defined element-wise by subtracting the multiplicity of each element in 
$B$ from its multiplicity in 
$A$ (but never going below zero). We write $a\geq\Omega(f)$ or $a\leq O(f)$ to mean that $a\geq c\cdot f$ or $a\leq c\cdot f$, respectively, for some global (unspecified) constant $c$.

\section{A Bayesian
analogue for Narcissus resiliency}\label{sec:vanilla}

As we mentioned in the introduction, in this work we advocate that in order to classify an attack as a ``successful reconstruction'' we must take into consideration the amount of side-information the attacker has (or lack thereof). Otherwise, it is unclear whether the attacker really extracted this reconstruction from the released model/statistics rather than from its side-information.
This is not captured by Definition~\ref{def:narc}. Thus, our starting point is to reformulate Definition~\ref{def:narc} using a Bayesian perspective, aiming to model the uncertainty the attacker has about the underlying data distribution.

\begin{definition}\label{def:vanilla}
Let $\XXX$ be a data domain and let $\Prior$ be a meta-distribution. That is, $\Prior$ is a distribution over distributions over datasets containing elements from $\XXX$. Let $R:\XXX^*\times\{0,1\}^*\rightarrow\{0,1\}$ be a reconstruction relation.
Algorithm $\MMM$ is {\em $(\eps,\delta,\Prior,R)$-\BN} if for all attackers $\AAA$ it holds that
\begin{equation}\label{eq:narcVanilla}
\underset{\substack{\Nature\leftarrow\Prior\\S\leftarrow\Nature\\y\leftarrow\MMM(S)\\z\leftarrow\AAA(y)}}{\Pr}[R(S,z)=1]\leq e^{\eps}\cdot \underset{\substack{\Nature\leftarrow\Prior\\S\leftarrow\Nature\\T\leftarrow\Nature\\y\leftarrow\MMM(S)\\z\leftarrow\AAA(y)}}{\Pr}[R(T,z)=1]+\delta.
\end{equation}
\end{definition}

The interpretation is as follows. 
Look at the left hand side of Equation (\ref{eq:narcVanilla}).
First $\Nature$ is drawn.
This is a distribution over datasets drawn from the meta-distribution $\Prior$.
Then we draw a dataset $S$ from $\Nature$.
The algorithm $\MMM$ trains a model $y$ on the dataset $S$, and finally the attacker $\AAA$ uses $y$ to reconstruct an element from $S$. The predicate $R$ determines if this reconstruction was successful or not. 
We say that $\MMM$ is {\em \BN} if the probability of a successful reconstruction from $S$ is not much larger than the probability of reconstructing an element from a fresh sample from $\Nature$ which we call $T$ on the right hand side of Equation (\ref{eq:narcVanilla}).

In this section, we construct an algorithm for the marginals problem (defined below) that is secure in the sense of Definition~\ref{def:vanilla} under interesting choices of the meta-distribution $\Prior$.

\begin{definition}
Let $n,d\in\N$ be parameters. Let $\MMM$ be an algorithm whose input is an $n\times d$ binary matrix. Algorithm $\MMM$ solves the $(n,d)$-marginals problem with error parameters $(\alpha,\beta)$ if for every input matrix $S$, with probability at least $1-\beta$, the algorithm returns a vector $y\in[0,1]^d$ satisfying $\| y - \Avg(S)\|_{\infty}\leq\alpha$, 
where $\Avg(S)\in[0,1]^d$ is the average of the rows of $S$.
\end{definition}

We show that the marginals problem can be securely solved (in the sense of Definition~\ref{def:vanilla}) under the following meta-distribution, which we refer to as the {\em Tardos-Prior}. This specific prior distribution is notoriously hard in the privacy and memorization literature, underlying many of the existing impossibility results.

\begin{definition}[Tardos-Prior]\label{def:Tardos}
Let $N\geq n\in\N$ and $d\in\N$ be parameters. We write $\Tardos(N,d,n)$ to denote the  meta-distribution defined by the following process. To sample a distribution form $\Tardos(N,d,n)$ do:
\begin{enumerate}[leftmargin=20px,topsep=0px,itemsep=-2px]
    \item Sample parameters $p_1,\dots,p_d\in[0,1]$ uniformly and independently.
    \item Generate an $N\times d\in\{0,1\}^{N\times d}$ matrix $C$ where each entry $C[i,j]$ is sampled independently from $\Ber(p_j)$.
    \item\label{step:UnifTardos} Return the uniform distribution over $S\subseteq C$ of size $n$. That is, the distribution of sampling $n$ rows from $C$ without repetitions.
\end{enumerate}

\end{definition}

We next show that the mechanism that returns the {\em exact average} is secure for Tardos-Prior. Note that we cannot hope to show this for every relation $R$. 
In particular the exact average is not secure w.r.t.\ the relation $R$ defined as $R(S,z)=1$ iff $z$ is the average of $S$. We show that it is secure for several natural choices for the relation $R$.

\subsection{Preventing near-exact reconstruction}

We start by considering attackers that aim to recover an element of $S$ to within a small Hamming distance. Specifically, we consider the following ``Hamming reconstruction relation'':

\begin{definition}
For $\gamma\in[0,1]$ let $R^H_{\gamma}$ be defined as follows. For a dataset $S\in(\{0,1\}^d)^n$ and an element $z\in\{0,1\}^d$ we have $R^H_{\gamma}(S,z)=1$ iff $\exists x\in S$ such that $\Hamming(x,z)\leq \gamma d$.
\end{definition}

Note that for $\gamma=0$ we have that $R^H_{\gamma}(S,z)=1$ iff $z\in S$.

\begin{lemma}\label{lem:averageVanilla}
Let $N\geq n\geq\Omega(1)$ and $d\geq\Omega(\log\frac{n}{\delta})$. 
The exact average is secure (in the sense of Definition~\ref{def:vanilla}) under Tardos-Prior with the relation $R^H_{\gamma}$ for every $\gamma\leq\frac{1}{25}$.
\end{lemma}

\begin{proof}
Let $\Nature\leftarrow\Tardos$ and let $S\leftarrow\Nature$ be a sample of size $n$. By the Chernoff bound, for $d\geq\Omega(\log\frac{1}{\delta})$ and $N,n\geq\Omega(1)$, with probability at least $1-\frac{\delta}{2}$, there are at least $d/5$ columns in $S$ whose Hamming weight is between $n/4$ and $3n/4$.\footnote{To see this, first observe that by the Chernoff bound, with probability at least $1-O(\delta)$, there are $\Omega(d)$ columns $j$ in $\Nature$ with $p_j\in\frac{1}{2}\pm\frac{1}{5}$.
When this is the case, assuming that $n\geq\Omega(1)$, then with probability at least $1-O(\delta)$, at least $\Omega(d)$ of these columns will have Hamming weight between $\frac{n}{4}$ and $\frac{3n}{4}$ in $S$.} 
We refer to such columns as {\em lukewarm}.
We proceed with the analysis assuming that there are at least $d/5$ lukewarm columns in $S$.

Let $y=\MMM(S)\in[0,1]^d$ denote the average of the rows in $S$. Note that conditioning on $y$ fixes the number of ones in every column in $S$, but the permutation in each column is still completely uniform and independent of the other columns. Therefore, for any choice of $z\leftarrow\AAA(y)$, by the Chernoff bound, assuming that $d\geq\Omega(\log\frac{n}{\delta})$, with probability\footnote{This probability is over choosing the permutation of each column in $S$.} at least $1-\frac{\delta}{2n}$, this $z$ agrees with the {\em first row} in $S$ on at most $\frac{4}{5}$-fraction of the lukewarm coordinates of $S$. As there are at least $d/5$ lukewarm coordinates, this means that the Hamming distance between $z$ and the first row in $S$ is at least $d/25$. By a union bound, the probability that such a $z$ is at Hamming distance at least $\frac{d}{25}$ from {\em all} the rows in $S$ is at least $1-\delta$. This shows that for any attacker $\AAA$ and for any $\gamma\leq\frac{1}{25}$ it holds that
$$
\underset{\substack{\Nature\leftarrow\Tardos\\S\leftarrow\Nature\\y\leftarrow\MMM(S)\\z\leftarrow\AAA(y)}}{\Pr}[R^H_{\gamma}(S,z)=1]\leq \delta.
$$
\end{proof}

Note that in this proof we did not leverage the RHS of Definition~\ref{def:vanilla}. One might hope that by taking the RHS into account, we might be able to avoid the requirement that $\gamma\leq\frac{1}{25}$. The (incorrect) idea is that  while the probability of identifying a vector $z$ that is $\gamma d$-close to $S$ gets larger when $\gamma$ increases, the RHS of Definition~\ref{def:vanilla} also increases. This however is not true, as for larger values of $\gamma$ the attacker could be extremely more likely to find a $z$ that is $\gamma d$-closer to $S$ than to $T$. We overcome this by considering a relaxed version of Definition~\ref{def:vanilla}.

\subsection{Bi-criteria variant of the definition}

Recall that Definition~\ref{def:vanilla} leverages a relation $R$ to determine if a reconstruction $z$ is valid, and the same relation $R$ is used in both sides of Inequality~(\ref{eq:narcVanilla}). This, however, could be considered too restrictive. For example, suppose that the dataset contains images, and there is an attacker that given the released model/statistics can reconstruct Bob's image to within 0.1\% error. Suppose also that without the release, matching this 0.1\% error is hard, but one can easily achieve an error of 0.11\%. So the attacker gains something from the release, but this gain is arguably quite weak. Should this be considered a successful reconstruction? 
We believe that both answers are valid. To accommodate a negative answer to this question, in the following definition we leverage {\em two} relations $R,\hat{R}$ instead of just one, in order to allow the criteria applied to the fresh dataset $T$ to be different from the one applied to $S$.

\begin{definition}\label{def:bi-vanilla}
Let $\XXX$ be a data domain and let $\Prior$ be a meta-distribution. That is, $\Prior$ is a distribution over distributions over datasets containing elements from $\XXX$. Let $R$ and $\hat{R}$ be reconstruction relations.
Algorithm $\MMM$ is {\em $(\eps,\delta,\Prior,R,\hat{R})$-\BN} if for all attackers $\AAA$ it holds that
\begin{equation}\label{eq:narc}
\underset{\substack{\Nature\leftarrow\Prior\\S\leftarrow\Nature\\y\leftarrow\MMM(S)\\z\leftarrow\AAA(y)}}{\Pr}[R(S,z)=1]\leq e^{\eps}\cdot \underset{\substack{\Nature\leftarrow\Prior\\S\leftarrow\Nature\\T\leftarrow\Nature\\y\leftarrow\MMM(S)\\z\leftarrow\AAA(y)}}{\Pr}[\hat{R}(T,z)=1]+\delta.
\end{equation}
\end{definition}

We show that the exact average is secure w.r.t.\ Tardos-Prior under the above definition. Specifically,

\begin{lemma}\label{lem:bi}
Let $N\geq 10n\geq\Omega(\log\frac{1}{\delta})$ and $d\geq\Omega(\log\frac{n}{\delta})$, and let $\gamma\in[0,1]$. 
The exact average is secure (in the sense of Definition~\ref{def:bi-vanilla}) for Tardos-Prior with the relations $R^H_{\gamma}$ and $R^H_{\hat{\gamma}}$ for $\hat{\gamma}=\left(1+O\left(  \sqrt{\frac{\ln\frac{d}{\delta}}{n}}  \right) \right) \gamma$. 
\end{lemma}

The proof is similar in spirit to the proof of Lemma~\ref{lem:averageVanilla}, but becomes much more technical as we need to argue about both sides of Inequality~(\ref{eq:narc}).

\begin{proof}[Proof of Lemma~\ref{lem:bi}]
By Lemma~\ref{lem:averageVanilla}, we already know that the exact average is secure whenever $\gamma\leq\frac{1}{25}$. We thus proceed with the analysis assuming that $\gamma\geq\frac{1}{25}$.

We consider the following process for sampling $S$ and $T$:
\begin{enumerate}
    \item Sample $p_1,\dots,p_d\in[0,1]$ uniformly and independently.
    \item Sample two sets of indices of size $n$ independently: $I_S\subseteq[N]$ and $I_T\subseteq[N]$.
    \item For every $i\in I_S\cup I_T$ sample $x_i\leftarrow\Ber(\vec{p})$ independently, where $\Ber(\vec{p})$ denotes the distribution over $d$-bit vectors where the $j$th bit is sampled independently from $\Ber(p_j)$.

    \item Define $D_{ST}=\{x_i : i\in I_S\cap I_T\}$, $D_{S}=\{x_i : i\in I_S\setminus I_T\}$, $D_{T}=\{x_i : i\in I_T\setminus I_S\}$.
    \item Return $S=D_S\cup D_{ST}$ and $T=D_T\cup D_{ST}$.
\end{enumerate}
Observe that the distribution on $(S,T)$ induced by this process is identical to the distribution induced by first sampling $\Nature\leftarrow\Prior$ and then sampling $S,T$ independently from $\Nature$.

For $j\in[d]$, let $p^{ST}_j,p_j^S,p^T_j$ denote the ``column probabilities'' in $D_{ST},D_S,S_T$, respectively. That is, $p^{ST}_j,p_j^S,p^T_j$ are the fraction of 1's in the $j$th column of $D_{ST},D_S,S_T$. We also denote $k=|I_S\cap I_T|$.  

Now fix $\vec{p},\vec{p}^{ST},\vec{p}^S,\vec{p}^T$, and $k$. This, in particular, fixes the average of $S$ and thus fixes the reconstructed element $z$.
We denote the elements of $D_S$ as $s_1,\dots,s_{n-k}$, the elements of $D_T$ as $t_1,\dots,t_{n-k}$, and the elements of $D_{ST}$ as $st_1,\dots,st_k$.

Let us start by calculating the probability that $\Hamming(z,s_1)\leq\gamma d$, where $s_1$ is the first row in $D_S$. 
    This Hamming distance is a random variable which is the sum of $d$ Bernoullis, where the parameter of the $j$th Bernoulli, denoted as $w^S_j$, is either $p^S_j$ or $1-p^S_j$, depending on $z[j]$. (The probability here is only over the permutation in each column of $D_S$.) Formally, let
    $$
    w^S_j= \begin{cases}
                        p^S_j   &,\quad\text{if } z[j]=0 \\
                        1-p^S_j\; &,\quad\text{if } z[j]=1
                    \end{cases}
    $$
We have that
\begin{align*}
\Pr[\Hamming(z,s_1)\leq\gamma d | \vec{p},\vec{p}^{ST},\vec{p}^S,\vec{p}^T,k]=\Pr\left[\sum_{j=1}^d \Ber\left( w^S_j  \right) \leq\gamma d \right]. 
\end{align*}
This holds analogously for elements of $D_T$ or $D_{ST}$ by replacing $w^S$ with $w^T$ or $w^{ST}$ (defined analogously to $w^S$).

 Next, the probability that $z$ matches {\em some} element in $S=D_S\cup D_{ST}$, i.e., the probability that $R^H_{\gamma}(S,z)=1$, can be written as
 {\small
\begin{align}
&\Pr[R^H_{\gamma}(S,z)=1 | \vec{p},\vec{p}^{ST},\vec{p}^S,\vec{p}^T,k]\nonumber\\
&=\Pr\left[\left. \left\{ R^H_{\gamma}(D_S,z)=1 \right\} \vee \left\{ R^H_{\gamma}(D_{ST},z)=1 \right\} \right| \vec{p},\vec{p}^{ST},\vec{p}^S,\vec{p}^T,k\right]
\nonumber\\
&=\Pr[\Hamming(z,s_1)\leq\gamma d \vee...\vee\Hamming(z,s_{n-k})\leq\gamma d
\vee \Hamming(z,st_1)\leq\gamma d\vee...\vee 
\Hamming(z,st_k)\leq\gamma d
| \vec{p},\vec{p}^{ST},\vec{p}^S,\vec{p}^T,k]\nonumber\\
&=\Pr\left[\left\{ \sum_{j=1}^d X^S_{1,j} \leq\gamma d  \right\}\vee...\vee \left\{ \sum_{j=1}^d X^S_{n-k,j} \leq\gamma d \right\}
\vee
\underbrace{
\left\{ \sum_{j=1}^d X^{ST}_{1,j} \leq\gamma d  \right\}\vee...\vee \left\{ \sum_{j=1}^d X^{ST}_{n-k,j} \leq\gamma d \right\}
}_{\text{call it Event ST}}
\right], \label{eq:SumBerS}
\end{align}
}
where each $X^S_{i,j}$ is a binary RV taking the value $1$ iff $z[j]\neq s_i[j]$,
and each $X^{ST}_{i,j}$ is a binary RV taking the value $1$ iff $z[j]\neq st_i[j]$. 
These Bernoullis are over the permutation within each column of $D_S$ or $S_{ST}$. Note that the RV's within each summation are independent. (However, these RV's  are {\em not} independent across different summations.) Also note that for every $i,j$, the marginal distribution over $X^S_{i,j}$ and $X^{ST}_{i,j}$ is $\Ber(w^S_j)$ and $\Ber(w^{ST}_j)$, respectively.

\medskip

Analogously, for $T=D_T\cup D_{ST}$ we have 
{\small
\begin{align}
&\Pr[R^H_{(1+\Delta)\gamma}(T,z)=1 | \vec{p},\vec{p}^{ST},\vec{p}^S,\vec{p}^T,k]\nonumber\\
&=\Pr\left[\left.\left\{R^H_{(1+\Delta)\gamma}(D_T,z)=1 \right\} \vee \left\{ R^H_{(1+\Delta)\gamma}(D_{ST},z)=1\right\} \right| \vec{p},\vec{p}^{ST},\vec{p}^S,\vec{p}^T,k\right]\nonumber\\
&\geq\Pr\left[\left.\left\{R^H_{(1+\Delta)\gamma}(D_T,z)=1 \right\} \vee \left\{ R^H_{\gamma}(D_{ST},z)=1\right\} \right| \vec{p},\vec{p}^{ST},\vec{p}^S,\vec{p}^T,k\right]\nonumber\\
&=\Pr\left[\left\{ \sum_{j=1}^d X^T_{1,j} \leq (1{+}\Delta)\gamma d  \right\}{\vee}...{\vee} \left\{ \sum_{j=1}^d X^T_{n-k,j} \leq (1{+}\Delta)\gamma d \right\}
{\vee}
\underbrace{
\left\{ \sum_{j=1}^d X^{ST}_{1,j} \leq\gamma d  \right\}{\vee}...{\vee} \left\{ \sum_{j=1}^d X^{ST}_{n-k,j} \leq\gamma d \right\}
}_{\text{Event ST}}
\right], \label{eq:SumBerT}
\end{align}
}
where each $X^T_{i,j}$ is a binary RV taking the value $1$ iff $z[j]\neq t_i[j]$. 
To connect between the expressions given in~(\ref{eq:SumBerS}) and~(\ref{eq:SumBerT}), let us consider the following two good events:

\begin{align*}
\text{Event } G_1:=& \text{ For every } j\in[d] \text{ it holds that } |p_j-p^S_j|\leq O\left(\sqrt{\frac{\ln\frac{d}{\delta}}{n}}\right)
    \text{ and }
    |p_j-p^T_j|\leq O\left(\sqrt{\frac{\ln\frac{d}{\delta}}{n}}\right).\\
\text{Event } G_2:=&
\text{ For every } i\in[n-k] \text{ it holds that}
\sum_{j=1}^d \left( X^T_{i,j} - X^S_{i,j} \right) \leq \Delta\gamma d.
\end{align*}

We will show that $G_1$ happens w.p.\ $1-O(\delta)$, and when this happens then $G_2$ happens w.p.\ $1-O(\delta)$. So overall $G_2$ happens w.o.p. Note that conditioned on $G_2$, for every $i$ we have that the event $\left\{ \sum_{j=1}^d X^S_{i,j} \leq\gamma d  \right\}$ is contained in the event $\left\{ \sum_{j=1}^d X^T_{1,j} \leq(1{+}\Delta)\gamma d   \right\}$, in which case 
the expression given in~(\ref{eq:SumBerT}) upper bounds the one in~(\ref{eq:SumBerS}), which basically implies the lemma. We will state this formally after showing that $G_1,G_2$ occur w.o.p. 

For $G_1$, note that as $N\geq 10n$ then by the Chernoff bound, for $n\geq\Omega(\log\frac{1}{\delta})$, with probability at least $1-\delta$, we have that $|D_S|=|D_T|\geq\frac{n}{2}$. When this is the case, then Event $G_1$ happens with  probability at least $1-O(\delta)$ by the Hoeffding bound (and a union bound). So $\Pr[G_1]\geq1-O(\delta)$.

We now show that for any fixture of $\vec{p},\vec{p}^{ST},\vec{p}^S,\vec{p}^T,k$ such that $G_1$ holds, then $G_2$ happens with probability at least $1-O(\delta)$. By a union bound, it suffices to show that 
$$\Pr\left[\left\{\sum_{j=1}^d \left( X^T_{1,j} - X^S_{1,j} \right) \leq \Delta\gamma d\right\}\right]\geq1-O\left(\frac{\delta}{n}\right).$$

This is a sum of $d$ independent RV's, each taking values in $\{-1,0,1\}$, each with expectation 
$
\E\left[X^T_{1,j}-X^S_{1,j}\right]=w^T_j-w^S_j\leq O\left(\sqrt{\frac{\ln\frac{d}{\delta}}{n}}\right)
$. This is because we fixed $\vec{p},\vec{p}^{ST},\vec{p}^S,\vec{p}^T,k$ such that $G_1$ holds (this also fixes $z$ and so $w_j,\hat{w}_j$ are well defined in terms of $\vec{p^S},\vec{p^T}$). 
Recall that we assumed that $\gamma\geq\frac{1}{25}$, and 
let $\Delta=\frac{1}{\gamma} \cdot \Theta\left(\sqrt{\frac{\ln\frac{d}{\delta}}{n}}\right)=\Theta\left(\sqrt{\frac{\ln\frac{d}{\delta}}{n}}  \right)$. Asserting that $d\geq\Omega(\ln\frac{n}{\delta})$,
by the Hoeffding bound we have that
$$\Pr\left[\sum_{j=1}^d \left( X^T_{1,j} - X^S_{1,j} \right) \leq \Delta\gamma d \right]\geq1-O\left(\frac{\delta}{n}\right),$$
as required.
Overall, Event $G_1$ happens with probability at least $1-O(\delta)$, and then Event $G_2$ happens with probability at least $1-O(\delta)$.

\medskip

Finally,
{\small 
\begin{align*}
&\Pr[R^H_{\gamma}(S,z)=1]
=\sum_{\substack{\vec{p},\vec{p}^{ST},\\\vec{p}^S,\vec{p}^T,k}}\Pr\left[
\begin{matrix}
\vec{p},\vec{p}^{ST},\\
\vec{p}^S,\vec{p}^T,k
\end{matrix}
\right]\Pr\left[R^H_{\gamma}(S,z)=1\left| \vec{p},\vec{p}^{ST},\vec{p}^S,\vec{p}^T,k\right.\right]
\\
&=\sum_{\substack{\vec{p},\vec{p}^{ST},\\\vec{p}^S,\vec{p}^T,k}}\Pr\left[\begin{matrix}
\vec{p},\vec{p}^{ST},\\
\vec{p}^S,\vec{p}^T,k
\end{matrix}\right]\Pr\left[\left\{ \sum_{j=1}^d X^S_{1,j} \leq\gamma d  \right\}\;\vee\;\dots\;\vee\; \left\{ \sum_{j=1}^d X^S_{n,j} \leq\gamma d \right\} \vee \text{Event ST}\right]\\
&\leq O(\delta)+\sum_{\substack{\vec{p},\vec{p}^{ST},\\\vec{p}^S,\vec{p}^T,k\\\text{s.t.\ }G_1\text{ holds}}}\Pr\left[\begin{matrix}
\vec{p},\vec{p}^{ST},\\
\vec{p}^S,\vec{p}^T,k
\end{matrix}\right]\Pr\left[\left\{ \sum_{j=1}^d X^S_{1,j} \leq\gamma d  \right\}\;\vee\;\dots\;\vee\; \left\{ \sum_{j=1}^d X^S_{n,j} \leq\gamma d \right\} 
\vee \text{Event ST}
\right]\\
&\leq O(\delta)+\sum_{\substack{\vec{p},\vec{p}^{ST},\\\vec{p}^S,\vec{p}^T,k\\\text{s.t.\ }G_1\text{ holds}}}\Pr\left[\begin{matrix}
\vec{p},\vec{p}^{ST},\\
\vec{p}^S,\vec{p}^T,k
\end{matrix}\right]\Pr\left[G_2\wedge\left(\left\{ \sum_{j=1}^d X^S_{1,j} \leq\gamma d  \right\}\;\vee\;\dots\;\vee\; \left\{ \sum_{j=1}^d X^S_{n,j} \leq\gamma d \right\} 
\vee \text{Event ST}
\right)\right]\\
&\leq O(\delta)+
\hspace{-15px}
\sum_{\substack{\vec{p},\vec{p}^{ST},\\\vec{p}^S,\vec{p}^T,k\\\text{s.t.\ }G_1\text{ holds}}}\Pr\left[\begin{matrix}
\vec{p},\vec{p}^{ST},\\
\vec{p}^S,\vec{p}^T,k
\end{matrix}\right]\Pr\left[G_2{\wedge}\left(\left\{ \sum_{j=1}^d X^T_{1,j} \leq(1{+}\Delta)\gamma d  \right\}{\vee}...{\vee} \left\{ \sum_{j=1}^d X^T_{n,j} \leq(1{+}\Delta)\gamma d \right\} 
{\vee} \text{Event ST}
\right)\right]\\
&\leq O(\delta)+\sum_{\substack{\vec{p},\vec{p}^{ST},\\\vec{p}^S,\vec{p}^T,k\\}}\Pr\left[\begin{matrix}
\vec{p},\vec{p}^{ST},\\
\vec{p}^S,\vec{p}^T,k
\end{matrix}\right]\Pr\left[\left\{ \sum_{j=1}^d X^T_{1,j} \leq(1{+}\Delta)\gamma d  \right\}\vee...\vee \left\{ \sum_{j=1}^d X^T_{n,j} \leq(1{+}\Delta)\gamma d \right\} 
\vee \text{Event ST}
\right]\\
&=O(\delta)+\sum_{\substack{\vec{p},\vec{p}^{ST},\\\vec{p}^S,\vec{p}^T,k}}\Pr\left[
\begin{matrix}
\vec{p},\vec{p}^{ST},\\
\vec{p}^S,\vec{p}^T,k
\end{matrix}
\right]\Pr\left[R^H_{(1{+}\Delta)\gamma}(T,z)=1\left| \vec{p},\vec{p}^{ST},\vec{p}^S,\vec{p}^T,k\right.\right]\\
&=O(\delta)+\Pr\left[R^H_{(1{+}\Delta)\gamma}(T,z)=1\right].
\end{align*}
}
\end{proof}

\section{Incorporating side-information into the definition}\label{sec:sideInfo}
In the definitions presented in the previous section, the attacker's knowledge (beyond the release $y$) is captured solely by the meta-distribution $\Prior$. However, it is reasonable to assume that the attacker might possess additional information: Either additional knowledge about the dataset $S$ or about the specific target distribution $\Nature$. For example, the attacker might know a point from $S$. 
In this section we extend our definition to support such side-information. 

\subsection{Side-information about the sample}

We start by examining one natural way to extend the definition to capture side-information about the sample $S$ (but not yet about the target distribution $\Nature$).

\begin{definition}\label{def:narcSideS}
Let $\XXX$ be a data domain and let $\Prior$ be a meta-distribution. That is, $\Prior$ is a distribution over distributions over datasets containing elements from $\XXX$. Let $R:\XXX^*\times\{0,1\}^*\rightarrow\{0,1\}$ be a reconstruction relation.
Algorithm $\MMM$ is {\em $(\eps,\delta,\Prior,R)$-\BN} if for every attacker $\AAA$ and every side information function $\SideInfo$ it holds that
\begin{equation}
\underset{\substack{\Nature\leftarrow\Prior\\S\leftarrow\Nature\\
K\leftarrow\SideInfo(S)\\
y\leftarrow\MMM(S)\\z\leftarrow\AAA(y,K)}}{\Pr}[R(S,z)=1]\leq e^{\eps}\cdot \underset{\substack{\Nature\leftarrow\Prior\\S\leftarrow\Nature\\
K\leftarrow\SideInfo(S)\\T\leftarrow\Nature|K\\y\leftarrow\MMM(S)\\z\leftarrow\AAA(y,K)}}{\Pr}[R(T,z)=1]+\delta. \label{eq:narcSideS}
\end{equation}
\end{definition}

In words, in the above definition we allow the attacker to have arbitrary side-information about the sample $S$, captured by the function $K\leftarrow\SideInfo(S)$. We then sample the ``fresh'' (baseline) dataset $T$ from the conditional distribution $\Nature|K$. To motivate this, suppose for example that the attacker knows the first row in $S$. So if $S=(x_1,\dots,x_n)$ then $K=x_1$. Then, when sampling $T\leftarrow\Nature|K$ we have that $x_1$ is also included in $T$ and therefore if the attacker outputs $x_1$ as the recovered element then this would not be considered a valid reconstruction. Indeed, in order to break Inequality~(\ref{eq:narcSideS}), and thus break the security of $\MMM$, the attacker's goal is to succeed in the left experiment while failing in the right experiment.

We now show that, unlike before, once we allow the attacker to have side-information about the sample, then the exact average is no longer secure w.r.t.\ Tardos-Prior. Specifically,

\begin{observation}
    The exact average is not secure (according to Definition~\ref{def:narcSideS}) for Tardos-Prior with the relation $R^H_0$ (i.e., when the attacker's goal is to identify an element from $S$).
\end{observation}

\begin{proof}
Let $\SideInfo$ be the function that takes a dataset $S=(x_1,\dots,x_n)$ and returns the first $n-1$ elements in it. Let $\AAA$ be the attacker that given $y=\average(S)$ and $K=\SideInfo(S)=(x_1,\dots,x_{n-1})$, returns $z=x_n$. That is, the attacker recovers the last element in $S$ (exactly) by subtracting $x_1,\dots,x_{n-1}$ from the unnormalized average. So,
$
\underset{\substack{\Nature\leftarrow\Prior\\S\leftarrow\Nature\\
K\leftarrow\SideInfo(S)\\
y\leftarrow\MMM(S)\\z\leftarrow\AAA(y,K)}}{\Pr}[z\in S]=1.
$

Now suppose that we resample $T\leftarrow\Nature|K$, and denote $T=(x'_1,\dots,x'_n)$. Note that by sampling from $\Nature|K$ we ensure that the first $n-1$ elements in $T$ are identical to their corresponding elements in $S$, while $x'_n$ is fresh. The event that $z=x_n\in T$ is exactly the event that $x_n\in\{x_1,\dots,x_{n-1},x'_n\}$, which is very unlikely under Tardos-Prior. Specifically, 
\begin{align*}
\underset{\substack{\Nature\leftarrow\Prior\\S\leftarrow\Nature\\
K\leftarrow\SideInfo(S)\\T\leftarrow\Nature|K\\y\leftarrow\MMM(S)\\z\leftarrow\AAA(y,K)}}{\Pr}[z\in T]
&=
\underset{\substack{\Nature\leftarrow\Prior\\(x_1,\dots,x_n,x'_n)\leftarrow\Nature}}{\Pr}[x_n\in \{x_1,\dots,x_{n-1},x'_n\}]\leq\frac{1}{N} + n\cdot 2^{-\Omega(d)}.
\end{align*}
To see this, recall that $S=(x_1,\dots,x_n)$ are sampled without repetitions from $\Nature$. Thus, $\Pr[x_n\in\{x_1,\dots,x_{n-1}\}]\leq (n-1)\cdot 2^{-\Omega(d)}$. As for $x'_n$, with probability $1/N$ we have that $x'_n$ and $x_n$ are the same element from $\Nature$, and otherwise they are equal with probability at most $2^{-\Omega(d)}$.
\end{proof}

\subsection{An alternative treatment for the attacker's side-information}

So the exact average is not secure w.r.t.\ Tardos-Prior when the attacker possesses enough side information. How about a noisy average? Consider the mechanism $\MMM$ that takes a dataset $S\in(\{0,1\}^d)^n$ and returns $y=\average(S)+\noise$. That is, $\MMM$ adds independent noise to each of the coordinates of the empirical average of $S$. Now consider an attacker with the following side-information: A pair of vectors $(k_1,k_2)\in(\{0,1\}^d)^2$ such that one of them is a random element from $S$ and the other is a random element from $\Nature$, but we do not know which is which. Nevertheless, given the release $y=\average(S)+\noise$ the attacker can identify w.h.p.\ who among $k_1,k_2$ belongs to $S$, and return it as its reconstructed element.\footnote{More specifically, the attacker picks $k_b$ that maximizes $\langle k_b,y \rangle$. See, e.g., \cite{DworkSSUV15}.} On the other hand, in the ``baseline experiment'', given this side information, each of $k_1,k_2$ is exactly equally likely to be in $T$. Thus, the attacker would fail to correctly guess an element of $T$ with probability $\gtrsim 1/2$.\footnote{
Recall that in Definition~\ref{def:narcSideS} we sample $T$ from $\Nature$ conditioned on $K$. In our case, this means that one of $k_1,k_2$ is included in $T$, each w.p. $1/2$. Our attacker guesses one out of $k_1,k_2$ which is very likely to be in $S$, but the probability of this element being in $T$ is still $\approx1/2$. 
This probability is slightly larger than 1/2 because even for the element which is sampled independently of $T$ there is still a tiny (but non-zero) probability that this element will appear in $T$.} 

So the noisy average is not secure w.r.t.\ Tardos-Prior when the attacker possesses enough side information. However, one could argue that the attack in this example should {\em not} be considered successful, because the result of the attack (the reconstructed element) was entirely encoded in the attacker's side-information. In some sense, the attacker extracted $z$ more from its side-information than from the release $y$. 

This discussion suggests that a different treatment for the attacker's side-information might be needed. More specifically, we want a way to {\em measure} the amount of information that $K$ contains on $z$, and require that this quantity is ``not too high'' in order for the attack to be considered successful. One could try to quantify this is via standard tools from information theory, such as the {\em mutual information} between $K$ and $z$, or the {\em conditional entropy} of $z$ given $K$. Recall, however, that these are both {\em average} notions, that do not necessarily have meaningful implications on the {\em specific instantiations} of $K$ and $z$. That is, it could be that the mutual information between $K$ and $z$ (as random variables) is very low, but the specific instantiation of $K$ actually tells us quite a bit about $z$.

Motivated by this discussion, we  use a ``point wise'' variant of entropy (and conditional entropy), known as {\em information content}, {\em self-information}, or {\em surprisal}.

\begin{definition}[\citet{shannon1948,CoverThomas06}]\label{def:surprisal}
Let $X$ be a random variable and let $x$ be an outcome. The {\em surprisal} of $x$ w.r.t.\ $X$ is defined as
$h(x)=-\log\left(\Pr[X=x]\right)$. 
Given another random variable $Y$ and an outcome $y$, we also write
$h(x|y)=-\log\left(\Pr[X=x|Y=y]\right)$.
\end{definition}

Intuitively, this quantity measures the amount of ``surprise'' or ``information'' gained upon learning that the event $X=x$ has occurred: rare events (with small probability) carry more information, while common events carry less.

We now leverage this notion in order to address the relation between the attacker's side-information $K$ and the reconstructed element $z$ in Definition~\ref{def:narcSideS}. To simplify the presentation, we begin with a variant of our definition that is tailored to ``exact reconstruction'', where the attacker's goal is to recover an element from the sample exactly. We will later relax this assumption.

\begin{definition}\label{def:narcSideI}
Let $\XXX$ be a data domain and let $\Prior$ be a meta-distribution. That is, $\Prior$ is a distribution over distributions over datasets containing elements from $\XXX$.
Algorithm $\MMM$ is {\em $(\eps,\delta,\xi,\Prior)$-\BN} if for every attacker $\AAA$ and every side information function $\SideInfo$ it holds that
\begin{align}
\hspace{-10px}\Pr_{\substack{\Nature\leftarrow\Prior\\
S\leftarrow \Nature\\
K\leftarrow\SideInfo(\Nature,S)\\
y\leftarrow\MMM(S)\\
z\leftarrow \AAA(y,K)}}\hspace{-10px}\left[
\begin{array}{c}
z\in S, \quad \text{ and} \\[0.2em]
h\left(z|K,S{\setminus}\{z\}\right)\geq\xi
\cdot h\left(z|S{\setminus}\{z\}\right)
\end{array}
\right]
%
\leq
e^\epsilon \cdot \hspace{-15px}\Pr_{\substack{\Nature\leftarrow\Prior\\
S\leftarrow \Nature\\
K\leftarrow\SideInfo(\Nature,S)\\
T\leftarrow \Nature|K\\
y\leftarrow\MMM(S)\\
z\leftarrow \AAA(y,K)}}\hspace{-10px}\Big[z\in T\Big]+\delta,\label{eq:sideI}
\end{align}
where for a fixed element $z\in\XXX$, a fixed dataset $\hat{S}\in\XXX^{n-1}$, and a fixed value  $\hat{K}$ we have
$$h(z|\hat{K},\hat{S})\triangleq
  -
\log\left(\Pr_{\substack{\Nature\leftarrow\Prior\\
S\leftarrow \Nature\\
K\leftarrow\SideInfo(\Nature,S)\\
v\leftarrow S\\
S':=S\setminus\{v\}
}}\left[ v=z \left|
(K,S')=\left(\hat{K},\hat{S}\right)
\right.\right]\right) .$$ 
\end{definition}

In words, for the attacker $\AAA$ to be successful, it must identify an element $z \in S$ such that: (1) the probability that $z \in T$ for a fresh dataset $T$ is relatively small; and (2) the side-information $K$ does not reveal too much about $z$ beyond what all other elements in $S$ might reveal.

\begin{remark}
There are many variants of Definition~\ref{def:narcSideI} which seem to capture similar intuition. For example,
\begin{enumerate}
\item  In the definition above we chose to compare
$h(z|K,S{\setminus}\{z\})$ with
$h(z|S\setminus\{z\})$. This choice is somewhat arbitrary. One could definitely consider other expressions to compare between. Additionally, one could consider other conditioning instead of $(K,S\setminus\{z\})$ and $S\setminus\{z\}$. For example, one could consider $h(z|K,D)$ for a dataset $D$ sampled independently from $S$, possibly of different size.

\item One could require the event $\left\{h(z|K,S{\setminus}\{z\})\geq\xi
\cdot h(z|S\setminus\{z\})\right\}$ to hold in both sides of Inequality~\ref{eq:sideI}, rather than just on the left hand side (possibly with a different value for $\xi$ in the right hand side).

\end{enumerate}
\end{remark}

We next present a secure mechanism w.r.t.\ Definition~\ref{def:narcSideI} and Tardos-Prior. Specifically, we show that this can be achieved using {\em differential privacy (DP)} \cite{DMNS06}. Crucially, this only requires guaranteeing {\em per attribute DP}. That is, in our construction we do not need to pay in DP composition across the $d$ features. As a result, the size of the required dataset does not need to scale polynomially with $d$. This is in sharp contrast to known lower bound from the privacy literature where it is known that in order to resist membership inference attacks, under the same Tardos-Prior, the size of the dataset must grow polynomially with $d$. This demonstrates that, in some cases, protecting against reconstruction attacks can be significantly easier than protecting against privacy attacks.

Let us first  introduce some additional preliminaries from the literature on differential privacy. Informally, an algorithm that analyzes data satisfies differential privacy if it is robust in the sense that its
outcome distribution does not depend ``too much'' on any single data point. Formally,

\begin{definition}[\cite{DMNS06}]
Let $\MMM:\XXX^*\rightarrow Y$ be a randomized algorithm whose input is a dataset $S\in \XXX^*$. Algorithm $\AAA$ is {\em $(\eps,\delta)$-differentially private (DP)} if for any two datasets $S,S'$ that differ on one point (such datasets are called {\em neighboring}) and for any outcome set $F\subseteq Y$ it holds that
$
\Pr[\MMM(S)\in F]\leq e^{\eps}\cdot\Pr[\MMM(S')\in F]+\delta.
$
\end{definition}

One of the most basic techniques for designing differentially private algorithms is by 
injecting (carefully calibrated) random noise into the computation that obscures the effect of every single element in the dataset. In particular,

\begin{theorem}[\citet{DMNS06}]\label{thm:Lap}
Let $\eps>0$, and let $\MMM_{\eps}$ be the mechanism that on input $S=(x_1,\dots,x_n)\in\left([0,1]^d\right)^n$ adds independently generated noise with distribution $\Lap(\frac{1}{\eps n})$ to each of the $d$ coordinates of $\average(S)=\frac{1}{n}\sum_i x_i$. Then  $\MMM_{\eps}$ is $(\eps d,0)$-differentially private.
\end{theorem}

We are now ready to present a secure mechanism w.r.t.\ Definition~\ref{def:narcSideI}. Specifically, we show that algorithm $\MMM_{\eps}$ (returning the noisy average) is secure.

\begin{lemma}\label{lem:DPsecure}
Let $\delta=2^{-\Omega(d)}$, constant $\xi>0$, and $\eps=0$.
Then for $\hat{\eps}= O(1)$ (small enough constant) it holds that $\MMM_{\hat{\eps}}$ is secure (in the sense of Definition~\ref{def:narcSideI}) for Tardos-Prior with parameters $\eps,\delta,\xi$.
\end{lemma}

To prove this lemma, we need the following claim.

\begin{claim}\label{claim:luke-h}
Let $\hat{S}\in\XXX^{n-1}$ with at least $d/5$ lukewarm columns, and let $z\in\XXX$. Then for Tardos-Prior we have that
$$
h(z|\hat{S}):
=-\log\left(\Pr_{\substack{\Nature\leftarrow\Prior\\
S\leftarrow \Nature\\
v\leftarrow S\\
S':=S\setminus\{v\}}}[v=z|S'=\hat{S}]\right)=\Omega(d).
$$ 
Recall that a lukewarm column is a column whose Hamming weight is between $n/4$ and $3n/4$.
\end{claim}

We first prove Lemma~\ref{lem:DPsecure}. The proof of Claim~\ref{claim:luke-h} will follow.

\begin{proof}[Proof of Lemma~\ref{lem:DPsecure}]
Let $\Nature\leftarrow\Prior$ and let $S\leftarrow\Nature$. As in the proof of Lemma~\ref{lem:averageVanilla}, by the Chernoff bound, for $N\geq n=\Omega(1)$, with probability at least $1-2^{-\Omega(d)}$, there are at least $d/5$ \emph{lukewarm} columns in $S$, i.e.\ columns of Hamming weight  between $n/4$ and $3n/4$. %
By Claim~\ref{claim:luke-h}, for every such $S$ and every $z\in\XXX$ we have 
$h(z|S)= \Omega(d)$. 
We calculate 
\begin{align*}
&\Pr_{\substack{\Nature\leftarrow\Prior\\
S\leftarrow \Nature\\
K\leftarrow\SideInfo(\Nature,S)\\z\leftarrow\AAA(\MMM(S),K)}}[z\in S \text{ and } h(z|K,S{\setminus}\{z\})\geq\xi\cdot h(z|S\setminus\{z\})]\\
&\qquad\stackrel[\gray{(*1)}]{}{\leq} 2^{-\Omega(d)}+\Pr_{\substack{\Nature\leftarrow\Prior\\
S\leftarrow \Nature\\
K\leftarrow\SideInfo(\Nature,S)\\z\leftarrow\AAA(\MMM(S),K)}}[z\in S \text{ and } h(z|K,S{\setminus}\{z\})\geq\Omega(d)]\\
&\qquad  \stackrel[\gray{(*2)}]{}{\leq} 2^{-\Omega(d)}+n\cdot\hspace{-10px}\Pr_{\substack{\Nature\leftarrow\Prior\\
S\leftarrow \Nature\\
v\leftarrow S\\
K\leftarrow\SideInfo(\Nature,S)\\
z\leftarrow\AAA(\MMM(S),K)}}[z= v \text{ and } h(z|K,S{\setminus}\{v\})\geq\Omega(d)]\\
&\qquad \stackrel[\gray{(*3)}]{}{\leq} 2^{-\Omega(d)}+n\cdot e^{\hat{\eps} d}\cdot\hspace{-10px}\Pr_{\substack{\Nature\leftarrow\Prior\\
S\leftarrow \Nature\\
v\leftarrow S\\
K\leftarrow\SideInfo(\Nature,S)\\
z\leftarrow\AAA(\MMM(S\setminus\{v\}),K)}}[z= v \text{ and } h(z|K,S{\setminus}\{v\})\geq\Omega(d)]\\
&\qquad
\stackrel[\gray{(*4)}]{}{=}
2^{-\Omega(d)}+n\cdot e^{\hat{\eps} d}\cdot\hspace{-10px}
\sum_{\substack{K,z,S{\setminus}\{v\}:\\ h(z|K,S{\setminus}\{v\})\geq\Omega(d)}}\Pr[K,z,S{\setminus}\{v\}]
\cdot\Pr_{v}[v=z | K,z,S{\setminus}\{v\} ]\\
&\qquad  \stackrel[\gray{(*5)}]{}{=}    2^{-\Omega(d)}+n\cdot e^{\hat{\eps} d}\cdot\hspace{-10px}
\sum_{\substack{K,z,S{\setminus}\{v\}:\\ h(z|K,S{\setminus}\{v\})\geq\Omega(d)}}\Pr[K,z,S{\setminus}\{v\}]
\cdot\Pr_{v}[v=z | K,S{\setminus}\{v\} ]\\
&\qquad \stackrel[\gray{(*6)}]{}{\leq} 2^{-\Omega(d)}+n\cdot e^{\hat{\eps} d}\cdot\hspace{-10px}
\sum_{\substack{K,z,S{\setminus}\{v\}:\\ h(z|K,S{\setminus}\{v\})\geq\Omega(d)}}\Pr[K,z,S{\setminus}\{S_{\ell}\}]
\cdot2^{-\Omega(d)}\\
&\qquad \leq2^{-\Omega(d)}+ n \cdot e^{\hat{\eps} d} \cdot 2^{-\Omega(d)}\\
&\qquad \stackrel[\gray{(*7)}]{}{\leq} n\cdot 2^{-\Omega(d)}.
\end{align*}
Here $\gray{(*1)}$ follows from Claim~\ref{claim:luke-h}. For $\gray{(*2)}$, note that $\Pr[z=v]=\Pr[z\in S]\cdot\Pr[z=v|z\in S]\geq\frac{1}{n}\Pr[z\in S]$. Inequality $\gray{(*3)}$ follows from the fact that $\MMM$ is $(\hat{\eps} d,0)$-differentially private. Specifically, instead of applying $\MMM$ to $S$ we apply it to $S\setminus\{v\}$.
Equality $\gray{(*4)}$ follows from the law of total probability. 
Equality $\gray{(*5)}$ is true because $z$ is a function (only) of $K$ and $S\setminus\{v\}$.
Inequality $\gray{(*6)}$ holds as for every $K,z,S\setminus\{v\}$ in this summation we have that $h(z|K,S{\setminus}\{v\})\geq\Omega(d)$, and thus $\Pr_{v}[v=z | K,S{\setminus}\{v\} ]\leq 2^{-\Omega(d)}$.
Finally, Inequality $\gray{(*7)}$ holds for small enough $\hat{\eps}\leq O(1)$.
\end{proof}

In remains to prove Claim~\ref{claim:luke-h}.

\begin{proof}[Proof of Claim~\ref{claim:luke-h}]

Fix $\hat{S}\in\XXX^{n-1}$ with $d/5$ lukewarm columns and fix $z\in\XXX$. We use $\hat{S}[j]$ to denote the $j$th column of $\hat{S}$. 
Also, given $\vec{p}=(p_1,\dots,p_d)$ we write $\Ber(\vec{p})$ to denote the distribution over $d$-bit vectors, where the $j$th bit is sampled from $\Ber(p_j)$. We have that
\begin{align*}
&\Pr_{\substack{\Nature\leftarrow\Prior\\
S\leftarrow \Nature\\
v\leftarrow S\\
S':=S\setminus\{v\}
}}\left[ v=z \left|S'=\hat{S}\right.\right]
\stackrel[\gray{(*1)}]{}{=} \Pr_{\substack{
\vec{p}\leftarrow[0,1]^d\\
S'\leftarrow\Ber(\vec{p})^{n-1}\\
v\leftarrow\Ber(\vec{p})
}}\left[ v=z \left|S'=\hat{S}\right.\right]
\\
&\qquad= \prod_{j\in[d]}\Pr_{\substack{
p_j\leftarrow[0,1]\\
S'[j]\leftarrow\Ber(p_j)^{n-1}\\
v_j\leftarrow\Ber(p_j)
}}\left[ v_j=z_j \left| S'[j]=\hat{S}[j] \right.\right]\\
&\qquad= \prod_{\substack{j\in[d]:\\z_j=1}}\Pr_{\substack{
p_j\leftarrow[0,1]\\
S'[j]\leftarrow\Ber(p_j)^{n-1}\\
v_j\leftarrow\Ber(p_j)
}}\left[ v_j=1 \left| S'[j]=\hat{S}[j] \right.\right]
\cdot \prod_{\substack{j\in[d]:\\z_j=0}}\Pr_{\substack{
p_j\leftarrow[0,1]\\
S'[j]\leftarrow\Ber(p_j)^{n-1}\\
v_j\leftarrow\Ber(p_j)
}}\left[ v_j=0 \left| S'[j]=\hat{S}[j] \right.\right]
\\
&\qquad=\prod_{\substack{j\in[d]:\\z_j=1}}\E_{\substack{
p_j\leftarrow[0,1]\\
S'[j]\leftarrow\Ber(p_j)^{n-1}
}}\left[ p_j \left| S'[j]=\hat{S}[j] \right.\right]
\cdot
\prod_{\substack{j\in[d]:\\z_j=0}}\left(1-\E_{\substack{
p_j\leftarrow[0,1]\\
S'[j]\leftarrow\Ber(p_j)^{n-1}
}}\left[ p_j \left| S'[j]=\hat{S}[j] \right.\right]\right)\\
&\qquad\stackrel[\gray{(*2)}]{}{\leq} 2^{-\Omega(d)}.
\end{align*} 
Here $\gray{(*1)}$ follows from the assumption that $S$ is sampled without repetitions from $\Nature$. Thus, instead of sampling $S$ of size $n$ from $\Nature$, then sampling $v\leftarrow\Nature$, and then setting $S':=S\setminus\{v\}$, we could simply sample $S',v$ directly from $\Ber(\vec{p})$.
Finally, $\gray{(*2)}$ follows from the assumption that there are $\Omega(d)$ lukewarm columns in $\hat{S}$ and by applying Laplace's rule of succession (derived using the properties of the Beta distribution, see e.g.,  \citep{jaynes2003probability}). 

\end{proof}

\subsection{Extending Definition~\ref{def:narcSideI} beyond perfect reconstruction}

In Definition~\ref{def:narcSideI} we focused on {\em exact recovery}, i.e., on attackers that aim to perfectly recover an element $z$ such that $z\in S$. This was convenient because it allowed us to argue about the {\em surprisal} of this recovered element $z$. 
In this subsection we extend Definition~\ref{def:narcSideI} beyond exact recovery. We consider a restricted family of relations $R$ in which the attacker's goal is to reconstruct (part of) a row from $S$, unlike general relations which possibly capture global information about $S$.  We refer to these types of attacks as {\em targeted} reconstruction or {\em targeted} relation. Note that $R^H_{\gamma}$ is an example of a targeted relation. On the other hand, the relation $R$ that given $S,z$ returns 1 iff $z$ specifies the first 3 bits of every point in $S$ is {\em not} a targeted relation.

\begin{definition}
A {\em targeted} reconstruction relation is a function $R:\XXX\times\{0,1\}^*\rightarrow\{0,1\}$. For a dataset $S\in\XXX^n$ we abuse notation and write $R(S,z)=1$ iff $\exists x\in S$ such that $R(x,z)=1$.
\end{definition}

\begin{definition}\label{def:narcSideIext}
Let $\XXX$ be a data domain and let $\Prior$ be a meta-distribution. That is, $\Prior$ is a distribution over distributions over datasets containing elements from $\XXX$. Let $R$ and $\hat{R}$ be reconstruction relations where $R$ is {\em targeted}.
Algorithm $\MMM$ is {\em $(\eps,\delta,\xi,\Prior,R,\hat{R})$-\BN} if for every attacker $\AAA$ and every side information function $\SideInfo$ it holds that
\begin{align*}
\Pr_{\substack{\Nature\leftarrow\Prior\\
S\leftarrow \Nature\\
K\leftarrow\SideInfo(\Nature,S)\\
y\leftarrow\MMM(S)\\
z\leftarrow \AAA(y,K)}}\hspace{-10px}\left[
\begin{array}{c}
\exists x\in S \text{ s.t. } R(x,z)=1 \text{ and}\\
h(x|K,S{\setminus}\{x\})\geq\xi\cdot h(x|S{\setminus}\{x\})
\end{array}
\right]
\leq
e^\epsilon \cdot \hspace{-10px}\Pr_{\substack{\Nature\leftarrow\Prior\\
S\leftarrow \Nature\\
K\leftarrow\SideInfo(\Nature,S)\\
T\leftarrow \Nature|K\\
y\leftarrow\MMM(S)\\
z\leftarrow \AAA(y,K)}}\hspace{-10px}[\hat{R}(T,z)=1]+\delta
\end{align*}
\end{definition}

\bibliographystyle{plainnat}

\end{document}